\documentclass[a4paper,UKenglish]{lipics-v2018}

\usepackage{microtype}
\usepackage{mathtools,amsmath} \usepackage{amssymb,amsthm}
\usepackage{graphicx} \usepackage{xspace} \usepackage{xypic}
\usepackage[numbers]{natbib} \usepackage{hyperref,color,lineno}
\usepackage[linesnumbered,lined,ruled]{algorithm2e}

\theoremstyle{plain}
\newtheorem{observation}[theorem]{Observation}

\graphicspath{{figures/}}
\title{Feedback Vertex Set on Geometric Intersection Graphs%
}
\titlerunning{Feedback Vertex Set on Geometric Intersection Graphs
} 

\author{
	Shinwoo An}{POSTECH}{shinwooan@postech.ac.kr}{}{}
\author{Eunjin Oh}{POSTECH}{eunjin.oh@postech.ac.kr}{}{}

\authorrunning{
	S. An and E. Oh} 

\Copyright{
	Shinwoo An and Eunjin Oh}


\supplement{}

\funding{
	This work was supported by the National Research Foundation of Korea (NRF) grant
	funded by the Korea government (MSIT) (No.2020R1C1C1012742).}

\acknowledgements{}

\EventEditors{John Q. Open and Joan R. Acces} \EventNoEds{2}
\EventLongTitle{42nd Conference on Very Important Topics (CVIT 2016)}
\EventShortTitle{CVIT 2016} \EventAcronym{CVIT} \EventYear{2016}
\EventDate{December 24--27, 2016} \EventLocation{Little Whinging,
	United Kingdom} \EventLogo{} \SeriesVolume{42} \ArticleNo{23}
\nolinenumbers 
\hideLIPIcs  

\newcommand{\contract}[1]{{#1}_{\mathcal{P}}}
\newcommand{\contractt}[1]{\bar{{#1}}_{\mathcal{P}}}
\newcommand{\contractb}[1]{{#1}_{\mathcal{B}}}
\newcommand{\ver}[2]{\Vert {#1}{#2}\Vert}

\newcommand{\fvs}{\textsc{Feedback Vertex Set }}
\DeclareMathOperator{\poly}{poly}

\begin{document}
	\maketitle

	\begin{abstract}
		In this paper, we present an algorithm for computing a feedback vertex set
		of a unit disk graph of size $k$, if it exists, which runs in time 
		$2^{O(\sqrt{k})}(n+m)$, where $n$ and $m$ denote the numbers of vertices and edges, respectively. 
		This improves the $2^{O(\sqrt{k}\log k)}n^{O(1)}$-time algorithm
		for this problem on unit disk graphs by Fomin et al. [ICALP 2017]. 
		Moreover, our algorithm is optimal assuming the exponential-time hypothesis.  
		Also, our algorithm can be extended to handle geometric intersection graphs
		of similarly sized fat objects without increasing the running time. 
	\end{abstract}
	
	\section{Introduction}
	%
	%
	The \fvs problem is a classic and fundamental graph problem, which is one of Karp's 21 NP-complete problems. 
	Given an undirected graph $G=(V,E)$ and an integer $k$, the goal is to find
	a set $S$ of vertices of size $k$ such that every cycle of $G$ contains at least
	one vertex of $S$. In other words, this problem asks to find a set $S$ of vertices
	of size $k$ whose removal from $G$ results in a forest. 
	This problem has been studied extensively from the viewpoint of exact exponential-time algorithms~\cite{DBLP:conf/iwpec/FominGP06}, parameterized algorithms~\cite{cao2015feedback, KOCIUMAKA2014556}, and approximation algorithms~\cite{BECKER1996167}. 
	
	In this paper, we study the \fvs problem from the viewpoint of parameterized
	algorithms. When the parameter is the size $k$ of a feedback vertex set, 
	the best known parameterized algorithm takes $3.62^kn^{O(1)}$ time~\cite{KOCIUMAKA2014556}.
	On the other hand, it is known that no algorithm for \fvs runs in $2^{o(k)}n^{O(1)}$ time assuming the exponential-time hypothesis (ETH)~\cite{PBook}.
	For special classes of graphs such as planar graphs and $H$-minor-free graphs for any fixed $H$, there are $2^{O(\sqrt{k})}n^{O(1)}$-time 
	algorithms for the \fvs problem~\cite{demaine2005subexponential}.
	Moreover, for planar graphs, \fvs admits a linear kernel~\cite{BONAMY201625}.
	
	We present a subexponential-time algorithm for \fvs on \emph{geometric intersection graphs},
	which can be considered as a natural generalization of planar graphs. 
	Consider a set $F$ of geometric objects (for example, disks and polygons) in the plane. The \emph{intersection graph} $G[F]$ 
	for $F$ is defined as the undirected graph whose  
	vertices correspond to the objects in $F$ such that  
	two vertices are connected by an edge if and only if the two objects 
	corresponding to them intersect. 
	In the case that $F$ is a set of disks, its intersection graph is called
	a \emph{disk graph}, which has been studied extensively for various algorithmic 
	problems~\cite{cabello2015shortest,chan2019approximate,kaplan2018routing}. 
	It can be used as a model for broadcast networks: The disks of $F$ represent transmitter-receiver stations with the same
	transmission power.  
	A planar graph can be represented as a disk graph, and thus the class of disk 
	graphs is a generalization of the class of planar graphs. 
	
	\subparagraph{Previous Work.}
	Prior to our work, the best known algorithm for \fvs parameterized by the size $k$ of a feedback vertex set on unit 
	disk graphs takes $2^{O(\sqrt{k} \log k)}n^{O(1)}$ time~\cite{fomin_et_al:LIPIcs:2017:7393,DBLP:conf/soda/FominLS12}.
	Since the best known lower bound on the computation time is 
	$2^{o(\sqrt{n})}$ assuming the exponential-time hypothesis (ETH)~\cite{de2020framework}, it is a natural question if
	\fvs on unit disk graphs can be solved optimally. 
	De Berg et al.~\cite{de2020framework} presented an (non-parameterized) ETH-tight
	algorithm for this problem, which runs in $2^{O(\sqrt{n})}$ time. 
	However, it was not known if there is an ETH-tight parameterized algorithm
	for \fvs on unit disk graphs.
	
	Recently, several NP-complete problems have been studied for unit disk graphs (and geometric intersection graphs) from the viewpoint of parameterized algorithms, for example, the \textsc{Steiner Tree}, \textsc{Feedback Vertex Set}, \textsc{Vertex Cover}, \textsc{$k$-Path} and \textsc{Cycle Packing} problems~\cite{bhore_et_al:LIPIcs:2020:12260,fomin_et_al:LIPIcs:2017:7393,10.1007/11847250_14}.  
	In the case of \textsc{Vertex Cover}, the work by de Berg et al.~\cite{de2020framework} implies an ETH-tight
	parameterized algorithm. Also, in the case of \textsc{$k$-Path} problem, 
	Fomin et al.~\cite{DBLP:conf/compgeom/FominLP0Z20} presented an ETH-tight parameterized algorithm which runs in $2^{O(\sqrt{k})}O(n+m)$ time.
	However, for the other problems, there is a gap between the running time of the best known algorithms and the best known lower bounds.
	
	\subparagraph{Our Result.}
	In this paper, we present an ETH-tight parameterized algorithm
	for the \fvs problem on unit disk graphs, which runs in $2^{O(\sqrt{k})}(n+m)$ time, where $n$ and $m$ denote the numbers of vertices and edges, respectively. 
	This improves the $2^{O(\sqrt{k}\log k)}n^{O(1)}$-time algorithm
	for this problem on unit disk graphs by Fomin et al.~\cite{fomin_et_al:LIPIcs:2017:7393}. 
	Moreover, unlike the algorithm in~\cite{fomin_et_al:LIPIcs:2017:7393}, our algorithm works on the graph itself and do not require the geometric representation. 
	Also, our algorithm indeed handles geometric intersection graphs of similarly sized fat objects in the plane without increasing the running time, which will be defined in Section~\ref{sec:graph}.
	
	\section{Preliminaries}\label{sec:prelim}
	For a graph $G=(V,E)$, we let $V(G)$ and $E(G)$ denote the sets of vertices and edges, respectively.
	For a subset $S$ of $V$, we let $G[S]$ be the subgraph of $G$ induced by $S$.
	Also, we let $G-S$ be the subgraph of $G$ induced by $V-S$.  
	
	\subsection{Geometric Intersection Graphs}\label{sec:graph}
	For a constant $0<\alpha<1$, an object $o\subseteq \mathbb{R}^2$ is said to be $\alpha$-$\emph{fat}$ if there are two disks $B_1$ and $B_2$ with 
	$B_1\subseteq o\subseteq B_2$ such that the radius ratio of $B_1$ and $B_2$
	is at least $\alpha$.\footnote{An \emph{object} is a point set in the plane, which is not necessarily connected.}
	See Figure~\ref{fig:fat}(a). 
	For example, a disk is a $1$-fat object, and a square is a $1/{\sqrt{2}}$-fat object. 
	We call a set $F$ of $\alpha$-fat objects a \emph{similarly sized set} if 
	the ratio between the largest and smallest diameter of the objects in $F$ is bounded by a fixed constant $\gamma$. 
	For an $\alpha$-fat object $o$ of $F$, 
	there are two \emph{concentric} disks $B_3$ and $B_4$  with 
	$B_3\subseteq o\subseteq B_4$ such that the radius ratio of $B_3$ and $B_4$ is
	at least $\alpha/2$. We consider the center of $B_3$ and $B_4$ as the \emph{center} of $o$. 
	For convenience, we assume that the smallest diameter of the objects of $F$
	is one.

	\begin{figure}
		\centering
		\includegraphics[width=0.7\textwidth]{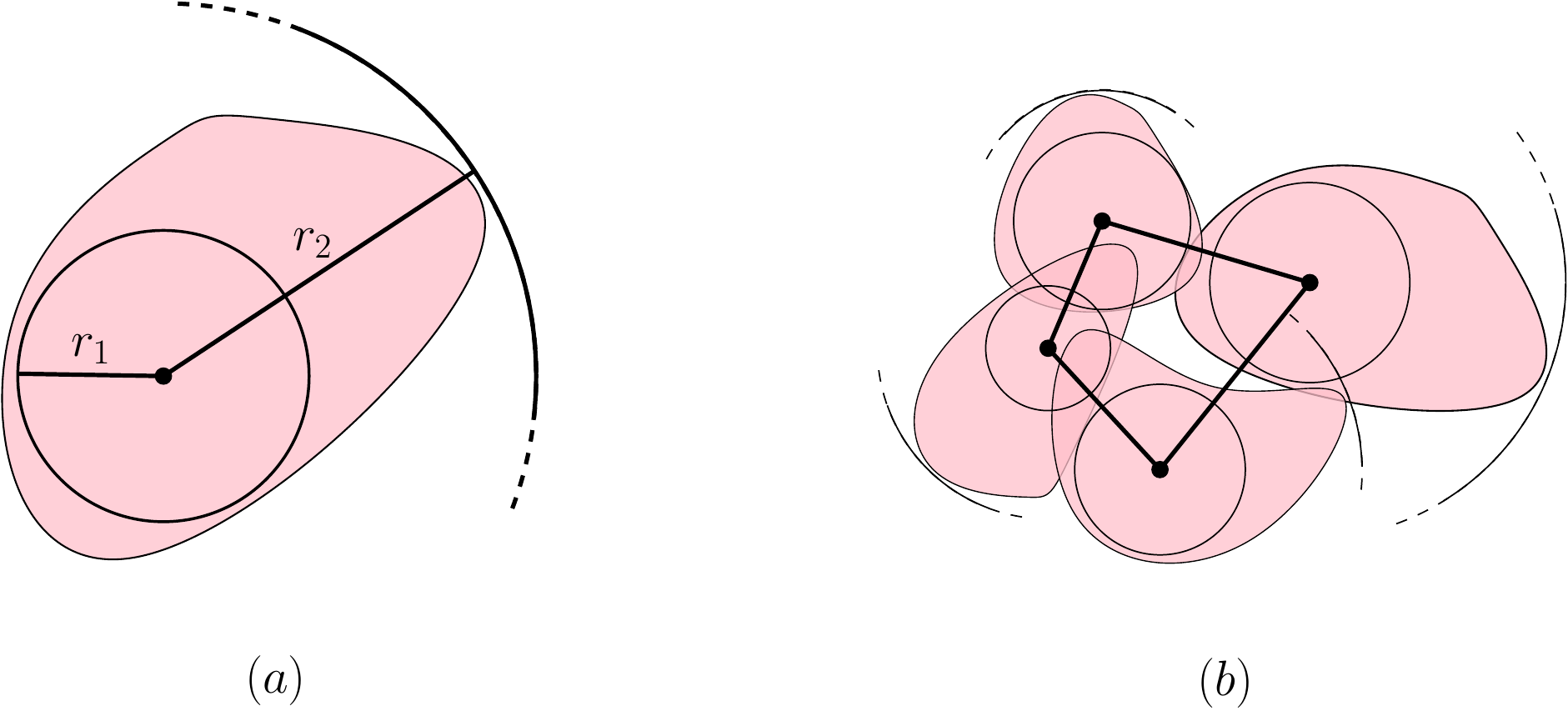}
		\caption{(a) A $\frac{r_1}{r_2}$-fat object. (b) The drawing of the geometric intersection graph $G[F]$.}
		\label{fig:fat}
	\end{figure}
	
	The \emph{intersection graph} $G[F]$ of a set $F$ of objects in $\mathbb{R}^2$ is 
	defined as the graph whose vertex set is $F$ and 
	two vertices are connected by and edge if and only if their corresponding objects intersect. 
	The \emph{drawing} of the intersection graph $G[F]$ is 
	a representation of $G[F]$ in the plane such that the vertices lie on the centers
	of the objects of $F$ and the edges are drawn as line segments. See Figure~\ref{fig:fat}(b).
	We sometimes use an intersection graph $G[F]$ and its drawing interchangeably if it is clear from the 
	context. 
	
	In this paper, we focus on geometric intersection graphs of objects in the \emph{plane} only.
	Because \fvs on unit ball graphs in $\mathbb{R}^3$ has no subexponential-time algorithm parameterized by $k$ unless ETH 
	fails~\cite{DBLP:conf/soda/FominLS12}.
	
	\subsection{Tree Decomposition and Weighted Width}
A \emph{tree decomposition} of a graph $G=(V,E)$ is defined as a pair $(T,\beta)$,
where $T$ is a tree and $\beta$ is a mapping from nodes of $T$ to subsets of $V$ (called bags) with the following properties.
Let $\mathcal {B}:=\{\beta(u) : u\in V(T)\}$ be the set of bags of $T$. 

\begin{enumerate}
	\item For any vertex $u\in V$, there is at least one bag in $\mathcal {B}$ which contains $u$;
	\item For any edge $(u,v)\in E$, there is at least one bag in $\mathcal {B}$ which contains both $u$ and $v$.
	\item For any vertex $u\in V$, the subset of bags of $\mathcal {B}$ containing $u$ forms a subtree of $T$.
\end{enumerate}

The $\emph{width}$ of a tree decomposition is defined as the size of its largest bag minus one, and the $\emph{treewidth}$ of $G$ is 
the minimum width of a tree decomposition of $G$.
In this paper, as in~\cite{de2020framework}, we use the notion of $\emph{weighted treewidth}$ introduced by~\cite{van2007safe}. 
Here, assume that each vertex $v$ has its weight.
The weight of each bag is defined as the sum of the weights of the vertices in the bag, and 
the \emph{weighted width} of tree decomposition is defined as the maximum weight of the bags. 
The \emph{weighted treewidth} of a graph $G$ is the minimum weighted width among the tree decompositions of $G$.
	
	\subsection{$\kappa$-Partition and $\mathcal{P}$-Contraction} 
	We use the concept of $\kappa$-$\emph{partitions}$ and $\mathcal{P}$-$\emph{contractions}$ introduced by de Berg et al.~\cite{de2020framework}. 
	Let $G[F]=(V,E)$ be a geometric intersection graph of a set $F$ of similarly sized $\alpha$-fat objects.
	A $\kappa$-$\emph{partition}$ of $G$ is a partition $\mathcal{P}$=$(P_1,...,P_t)$ of $V$ such that $G[P_i]$ is connected
	and is the union of  at most $\kappa$ cliques 
	for every partition class $P_i$. 
	Given a $\kappa$-partition $\mathcal{P}$, we consider a graph obtained by contracting the vertices in the same partition class to a single vertex and
	removing all 
	loops and parallel edges. 
	We call the resulting graph the $\mathcal{P}$-$\emph{contraction}$ of $G$ and denote it by $G_{\mathcal{P}}$.
	The weight of each vertex of $G_{\mathcal{P}}$ is defined as $\lceil \log |P_i| \rceil +1$, where $P_i$ denotes
	the partition class of $\mathcal{P}$ defining the vertex.
	
	De berg et al.~\cite{de2020framework} presented an $O(n+m)$-time algorithm
	for computing a $\kappa$-partition of $G[F]$ such that
	$G_{\mathcal P}$  has maximum degree at most $\delta$ for constants $\kappa$  and $\delta$.
	We call such a $\kappa$-partition a \emph{greedy partition} of $G$. 
	Given a greedy partition $\mathcal P$, 
	they showed that the weighted treewidth of $G_{\mathcal{P}}$ is $O(\sqrt{n})$.
	Moreover, they presented a $2^{O(\sqrt{n})}$-time algorithm for computing a tree decomposition of $G_{\mathcal P}$
	of weighted width $O(\sqrt{n})$.

	
	\subsection{Overview of Our Algorithm}
	Our algorithm consists of two steps: computing a (weighted) tree decomposition of $G$, and then
	using a dynamic programming on the tree decomposition. 
	As in~\cite{de2020framework}, we first compute 
	a greedy partition. Then we show that the weighted treewidth of $\contract{G}$ is $O(\sqrt{k})$ if $(G,k)$ is
	a \textbf{yes}-instance.  
	To do this, we first show that the weighted treewidth of $\contract{G}$
	is $O(\sqrt{k})$ if $G$ has $O(\sqrt{k})$ vertices of degree at least three in Section~\ref{sec:tw}. 
	Then we show that $G$ has $O(\sqrt{k})$ vertices of degree at least three
	in Section~\ref{sec:tw-yes}.
	
	Using this fact, we compute a constant approximation to the weighted treewidth of
	$\contract{G}$. If it is $\omega(\sqrt{k})$, we conclude that $(G,k)$ is a \textbf{no}-instance immediately. 
	Otherwise, we compute a feedback vertex set of size $k$, if it exists, using dynamic programming on a
	tree decomposition of $\contract{G}$ of weighted width $O(\sqrt{k})$. The dynamic programming algorithm is described in Section~\ref{sec:dp}.

	\section{Tree Decomposition and Weighted Treewidth}\label{sec:tw}
	In this section, we present the first step of our algorithm: computing a tree decomposition
	of $G_{\mathcal P}$ of weighted width $O(\sqrt{|U|})$ for a greedy partition $\mathcal P$ of a geometric intersection graph $G$, where 
	$U$ denotes the set of vertices of $G$ of degree at least three in $G$. 
	More precisely, we prove the following theorem. 
	
	\begin{theorem} \label{pc_wtw}
		Let $G$ be an intersection graph of similarly sized $\alpha$-fat objects
		with $n$ vertices and $m$ edges, and let $\mathcal{P}$ be a greedy partition of $G$. Then the weighted treewidth of $\contract{G}$ is $O(\sqrt{|U|})$,
		where $U$ denotes the set of vertices of $G$ of degree at least three in $G$. 
		Moreover,  we can compute a tree decomposition of $\contract{G}$ of weighted width $O(\sqrt{|U|})$
		in $2^{O(\sqrt {|U|})}(n+m)$ without using a geometric representation of $G$. 
	\end{theorem}
	
	In the following, 
	let $G$ be an intersection graph of similarly sized $\alpha$-fat objects, $\mathcal{P}$ be a greedy partition of $G$,
	and $U$ be the set of vertices of $G$ of degree at least three in $G$. 
	In Section~\ref{sec:width-proof}, we prove the first part of the theorem, and 
	in Section~\ref{sec:construction}, we prove the second part of the theorem. 
	
	\begin{figure}
		\centering
		\includegraphics[width=0.8\textwidth]{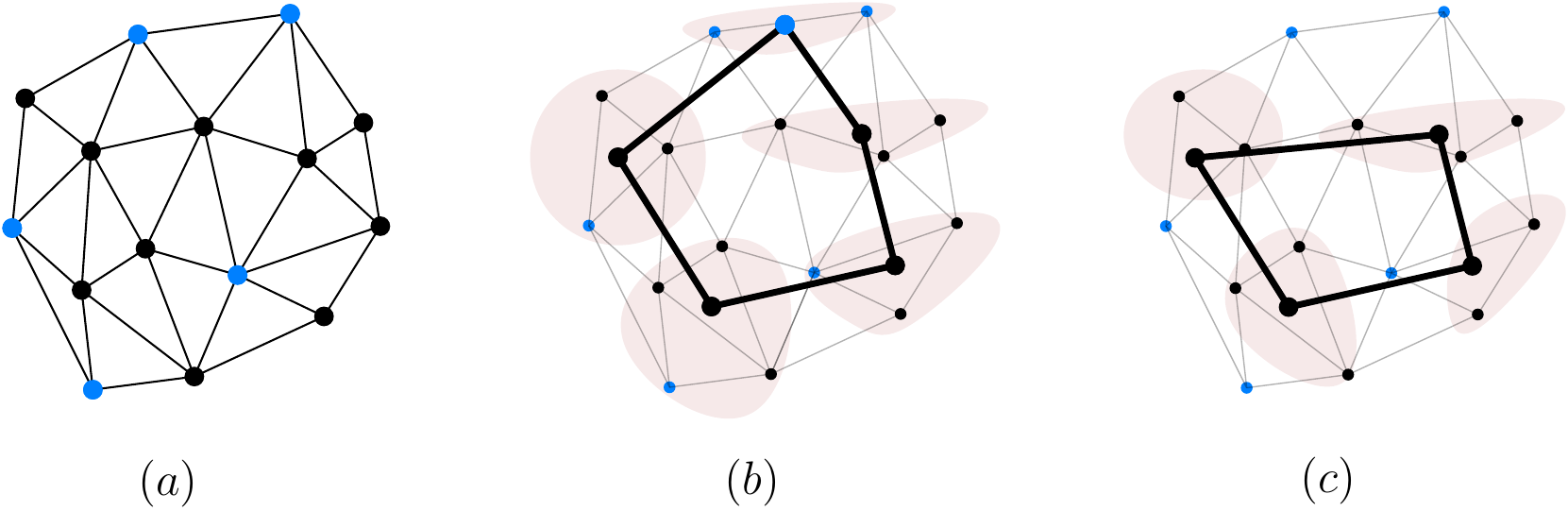}
		\caption{(a) The Delaunay triangulation $H$ of the vertex set of $V(G)$. The vertices of $G$ of degree at most two are colored blue.  
		(b) $\mathcal{P}$ partitions $V$ into five subsets. Each subset of $\mathcal{P}$
		is contained in a single pink region.
		Then the $\mathcal{P}$-contraction of $H$ is the cycle consisting of five edges.  (c)~Each subset of $\bar{\mathcal P}$ consists of the black points contained in a single pink region.}
		\label{fig:contraction}
	\end{figure}

	\subsection{Weighted Treewidth of the \texorpdfstring{$\mathcal P$}{P}-contraction}\label{sec:width-proof}
	We first show that the weighted treewidth of  $\contract{G}$ is $O(\sqrt{|U|})$. 
	To do this, we transform $G$ into a planar graph $H$. As we did for $G$, we compute the $\mathcal P$-contraction of $H$,
	and remove degree-2 vertices of the $\mathcal P$-contractions of $H$ and $G$ in a specific way.
	In this way, we obtain $\contractt{H}$ and $\contractt{G}$ of complexity $O(\sqrt{|U|})$. 
	Also, we assign the weight to each vertex of the resulting graphs. 
	Then we transform a balanced separator of $\contractt{H}$ of weight $O(\sqrt{|U|})$ into a balanced separator of  $\contractt{G}$ of weight $O(\sqrt{|U|})$. Using a weighted balanced separator, we compute 
	a tree decomposition of $\contractt{G}$ of weighted width $O(\sqrt{|U|})$, and then transform it into a tree decomposition of $\contract{G}$
	of weighted width $O(\sqrt{|U|})$. 
	
	\subparagraph{Delaunay triangulation and its contraction: $H$ and $\contract{H}$.}
	Consider the Delaunay triangulation $H$ of the point set $V(G)$. Here, we consider $V(G)$ as the set of the centers of the objects defining $G$. 
	We consider the Delaunay triangulation $H$ as an edge-weighted plane graph such that the \emph{length} of each edge is the Euclidean distance between their endpoints. 
	The Delaunay triangulation $H$ is a \emph{5.08}-\emph{spanner} of the complete Euclidean graph defined by $V(G)$~\cite{dobkin1990delaunay}.  
	That is, for any two points $u$ and $v$ of $V(G)$, the length of the shortest path 
	in $H$ between $u$ and $v$ is at most $5.08$ times their Euclidean distance. 

	Notice that $H$ might contain an edge which is not an edge of $G$. However, the following still holds as $G$ is a subgraph of the compete
	Euclidean graph. 
	For any edge $(u,v)$ in $G$, there is a $u$-$v$ path in $H$ consisting of 
	$5.08\|uv\|$ edges.
	Recall that $\alpha$ is the measure for the fatness of the objects defining $G$, which is a constant.  
	Let $\contract{H}$ be the $\mathcal P$-contraction of $H$. 
	Note that	$V(\contract{G})=V(\contract{H})$.
	For a vertex $v$ of $\contract{G}$ or $\contract{H}$, we let $P(v)$ to denote the partition class of $\mathcal{P}$ corresponding to $v$. 
	See Figure~\ref{fig:contraction}(b).
	
	
	\subparagraph{Partition $\mathcal {\bar P}$ of $U$: $\contractt{G}$ and $\contractt{H}$.}
	The size of $\contract{G}$ and $\contract{H}$ might be $\Theta(n)$ even if $|U|$ is small. 
	Recall that $U$ is the set of vertices of $G$ of degree at least three in $G$. 
	To obtain a balanced separator of $\contract{G}$ of small weight, we compute a new graph $\contractt{G}$ and $\contractt{H}$ of size $O(|U|)$ as follows.
	
	Let $\mathcal{P}=(P_1,\ldots, P_t)$, which is a partition of $V(G)$. 
	Using this, we consider the partition $\mathcal {\bar P} = (\bar P_1, \ldots \bar P_t)$ of $U$ such that 
	$\bar P_i = P_i \cap U$. Note that $\bar P_i=\emptyset$ if every vertex of $P_i$ has degree at most two in $G$. 
	We first ignore the empty partition classes from $\mathcal {\bar P}$, and let $\contractt{G}$ and $\contractt{H}$ be
	the $\mathcal {\bar P}$-contraction of $G$ and the $\mathcal {\bar P}$-contraction of $H$, respectively. 
	We add several edges to $\contractt{G}$ and $\contractt{H}$ by considering the empty partition classes of $\mathcal {\bar P}$ as follows. 
	Since $G[P_i]$ is connected and consists of $\kappa$ cliques, $G[P_i]$ is a simple cycle or a simple path if $\bar P_i=\emptyset$.
	Let $P_\emptyset$ be the union of $P_i$'s for all indices with $\bar P_i=\emptyset$. 
	Then each connected component of $G[P_\emptyset]$ is also a simple cycle or a simple path. 
	If the connected component is a simple cycle, no vertex in the component is connected by a vertex of $U$ in $G$. 
	Otherwise, the endpoints of the simple path, say $p$ and $q$, are contained in $U$. 
	In this case, we connect the vertices of $\contractt{G}$ (and $\contractt{H}$) by an edge unless this edge forms a loop. 
	See Figure~\ref{fig:contraction}(c).
	
	For a vertex $\bar v$ of $\contractt{G}$ or $\contractt{H}$, we let $\bar{P}(\bar v)$ be the partition class of $\mathcal{\bar P}$ corresponding to $\bar v$. 
	By construction, $\bar P(\bar v)$ is not empty for any vertex $\bar v$. 
	Each vertex $\bar v$ of $\contractt{G}$ has a $g$-\emph{weight} $g(\bar{v})=\lceil \log |\bar P(\bar v)| \rceil + 1$. 
	Also, each vertex $\bar v$ of $\contractt{H}$ has a $h$-\emph{weight} $h(\bar{v})=\sum_{\bar u \in N(\bar v)} (\lceil \log |\bar P(\bar u)|\rceil + 1)$, where
	$N(\bar v)$ is the set of vertices $\bar u$ of $\contractt{H}$ such that 
	there is a $\bar u$-$\bar v$ path in $\contractt{H}$ consisting of at most $t=(\delta+1)(10.16(\alpha\gamma)^{-1})^2\pi$ vertices.

	\begin{observation}
		A vertex $\bar v$ of $\contractt{H}$ has a constant degree, and thus 
		the size of $N(\bar v)$ is $O(1)$. 
	\end{observation}
	
	\subparagraph{Balanced separator of $\contractt{H}$ of small $h$-weights.}
	For a vertex-weighted graph $G'$ and a subgraph $H'$, 
	we denote the sum of the weights of the vertices of $H'$ by the $\emph{weight}$ of $H'$.
	a subset $S$ of $V(G')$ is called a \emph{balanced separator} of $G'$  
	if the weight of each connected component of $G'-S$ is at most $2/3$ of the weight of $G'$. 
	The \emph{weight} of a balanced separator $S$ of $G'$ is the sum of the weights of the vertices of $S$.  
	Djidjev~\cite{djidjev1997weighted} showed that  
	a planar graph $G'$ has a balanced separator of weight $O(\sqrt{\sum_{v\in V}w(v)^2})$, where $w(v)$ is the weight of a vertex $v$. 
	
	\medskip
	The following lemma implies that  $\contractt{H}$ has a balanced separator of $h$-weight $O(\sqrt{|U|})$.
	
	\begin{lemma}
		The sum of $h(\bar v)^2$ for all vertices $\bar v$ of $\contractt{H}$ is $O({|U|})$. 
	\end{lemma}
	\begin{proof}
		Let $\bar v$ be a vertex of $\contractt{H}$. The $h$-weight of $\bar v$ is the sum of $g(\bar u)= (\lceil \log |\bar P(\bar u)| \rceil+1)$ for all 
		vertices $\bar u$ of $N(\bar v)$ by definition. 
		By Cauchy–Schwarz inequality, $(  \sum_{\bar u\in N(\bar v)} g(\bar u))^2$ is at most 
		$|N(\bar v)| (\sum_{\bar u\in N(\bar v)} g(\bar u)^2)$. Since the size of $N(\bar v)$ is at most a constant, say $c$, 
		we have the following inequality. 
		\[
		{\sum_{\bar v \in V} h(\bar v)^2} \leq c \cdot \sum_{\bar v \in V} \sum_{\bar u\in N(\bar v)} g(\bar u)^2 =  c \cdot \sum_{\bar u \in V} \sum_{\bar v\in N(\bar u)} g(\bar u)^2,
		\]
		where $V$ denotes the vertex set of $\contractt{H}$. 
		Since $\sum_{\bar v\in N(\bar u)} g(\bar u)^2= |N(\bar u)|g(\bar u)^2 \leq c \cdot g(\bar u)^2$, we finally have the following. 

\[
\sum_{\bar v \in V} h(\bar v)^2 \leq c^2 \cdot \sum_{\bar u\in V}  g(\bar u)^2  = c^2 \cdot \sum_{\bar u\in V} (\lceil \log{|\bar P(\bar u)| \rceil}+1)^2 =
 O(\sum_{\bar u \in V}|\bar P(\bar u)|) = O(|U|).
\]
Here, the second equality holds because $2x \geq (\log x+1)^2$ for all values $x \geq 1$. 
Also, the last equality holds because $\mathcal {\bar P}$ is a partition of $U$. 
%
	\end{proof}
	
	\subparagraph{Balanced separator of $\contractt{G}$ of small $g$-weight.}
	Let $S$ be a balanced separator of $\contractt{H}$ of $h$-weight $O(\sqrt{|U|})$.
	Let $S'$ be the union of  $N(\bar v)$ for all $\bar v\in S$. 
	Recall that $|V(\contractt{G})|=|V(\contractt{H})|$. 
	The sum of the $g$-weights of the vertices of $S'$ is $O(\sqrt{|U|})$ by definition of the $g$- and $h$-weights. 
	Therefore, it suffices to show that the weight of each connected component of $\contractt{G}-S'$ is at most $2/3$ of the weight of $\contractt{G}$.
	
	\begin{figure}
		\centering
		\includegraphics[width=0.6\textwidth]{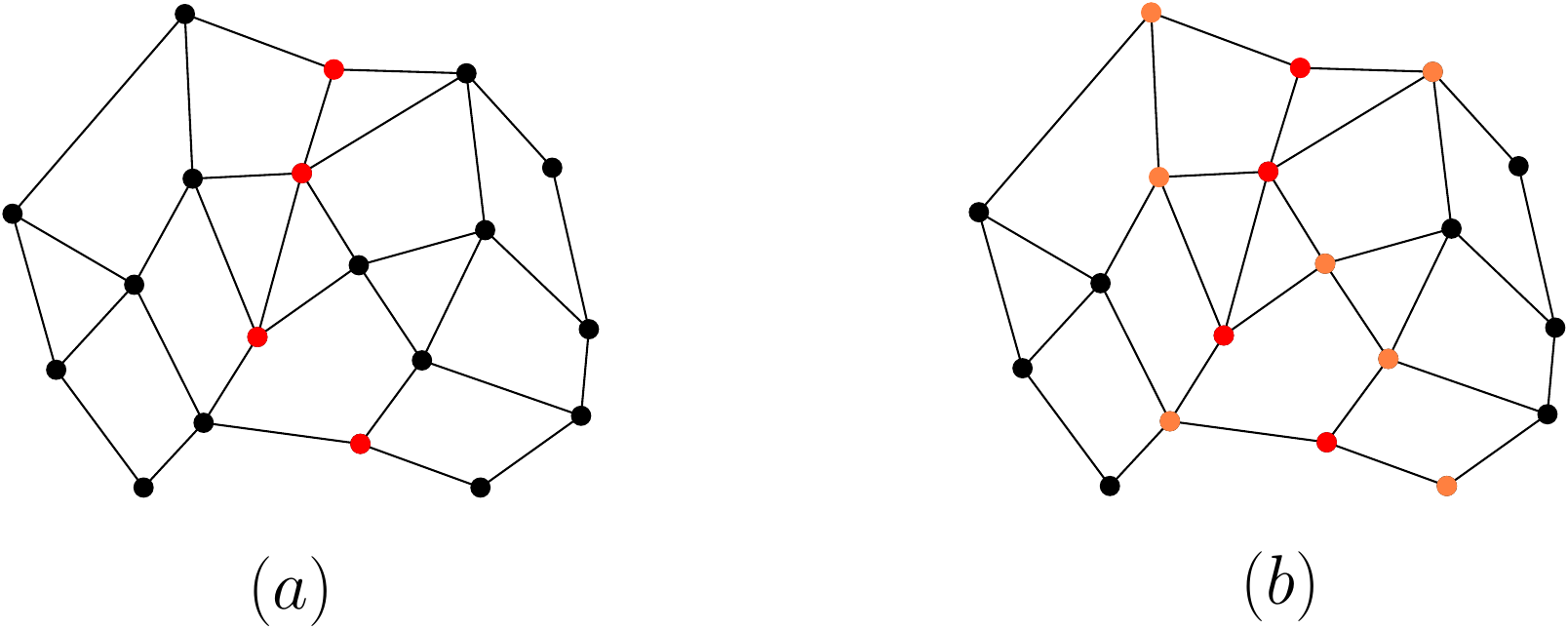}
		\caption{(a) A balanced separator $S$ of $\contractt{H}$ consists of the red vertices.  (b) A balanced separator $S'$ of $\contractt{G}$ 
			consists of the red vertices and the orange vertices.
		The vertices whose distance from the red vertices are at most $t=1$ are colored with orange.}
		\label{fig:separator}
	\end{figure}

	%
	\begin{lemma}\label{spanner}
		For any edge $(u,v)$ of $\contract{G}$, there is a $u$-$v$ path in $\contract{H}$ consisting of $t$ vertices.\footnote{Recall that $t=(\delta+1)(10.16(\alpha\gamma)^{-1})^2\pi$.}
	\end{lemma}
	\begin{proof}
		Consider an edge $(u,v)$ of $\contract{G}$.
		By construction, there is an edge $(p, q)$ in $G$ such that $p\in P(u)$ and $q\in P(v)$.
		Recall that Delaunay triangulation is a $5.08$-spanner of the complete Euclidean graph. 
		Therefore, there exists a $p$-$q$ path $\tau$ in $H$ such that the sum of the lengths (Euclidean distance between two endpoints)
		of the edges is at most $5.08 \ver{p}{q}$.
		
		We claim that $\tau$ intersects at most $(\delta+1)(10.16(\alpha\gamma)^{-1})^2\pi$ partition classes of $\mathcal{P}$, where $\delta$ is the maximum degree
		of $\contract{G}$, $\alpha$ is the measure for the fatness of the objects, and $\gamma$ is the ratio of the largest and smallest diameters of the objects. 
		Consider a $\emph{grid}$ of the plane, which is a partition of the plane into axis-parallel squares (called \emph{cells}) with diameter $1/\sqrt{2}$.
		Recall that the smallest diameter of the objects defining $G$ is one. 
		For a cell $\sigma$, we denote the set of vertices of $G$ contained in $\sigma$ by $P_\sigma$.
		By definition, $P_\sigma$ forms a clique in $G$. Since the maximum degree $\delta$ of $\contract{G}$ is constant, at most $\delta +1$ partition classes of $\mathcal P$ has their vertices in $P_\sigma$. 
		Note that every object is contained in a ball of radius $(\alpha\gamma)^{-1}$.
		Therefore, $\ver{p}{q}\leq 2(\alpha\gamma)^{-1}$. Then, the Euclidean length of $\tau$ is at most $10.16(\alpha\gamma)^{-1}$.
		This implies that $\tau$ is contained in a ball of radius $10.16(\alpha\gamma)^{-1}$ centered at $p$. In other words, $\tau$ intersects at most $(10.16(\alpha\gamma)^{-1})^2\pi$ cells.
		Since each cell intersects at most $\delta+1$ partition classes of $\mathcal P$, 
		$\tau$ intersects $(\delta+1)(10.16(\alpha\gamma)^{-1})^2\pi$ partition classes. Let $t=(\delta+1)(10.16(\alpha\gamma)^{-1})^2\pi$. 
		
		Then we consider a $u$-$v$ path in $\contract{H}$ obtained by replacing each point in $\tau$ to its contracted point in $\contract{H}$.
		This path now consists of $t$ different vertices of $\contract{H}$, but it is not necessarily simple. 
		We can obtain a simple path by removing  duplicated subpaths that have same endpoints. 
		Then we obtain a simple path of consisting of $t$ vertices. This completes the proof.
	\end{proof}
	
	\begin{lemma} \label{spanner2}
		For any edge $(\bar u,\bar v)$ of $\contractt{G}$, there is a $\bar u$-$\bar v$ path in $\contractt{H}$ consisting of $t$ vertices.
	\end{lemma}
	\begin{proof}
		We consider an edge $(\bar u,\bar v)$ of $\contractt{G}$.
		Let $u$ and $v$ be the vertices of $\contract{G}$ with $P(u)\cap U=\bar P(\bar u)$ and $P(v)\cap U=\bar P(\bar v)$. 
		By construction, either 
		there is an edge $(u,v)$ in $\contract{G}$, or there is a $u$-$v$ path
		in $\contract{G}$ such that no internal vertex is contained in $U$.
		We consider these two cases separately. 
		
		\begin{itemize}
			\item \textbf{Case 1. $(u,v)$ is an edge of $\contract{G}$.}
			In this case, there is a $u$-$v$ path $\tau_1$ in $\contract{H}$ consisting of $t$ vertices by Lemma~\ref{spanner}. 
			Let $I$ be the sequence of the indices of the partition classes of $\mathcal{P}$ corresponding to the vertices of $\tau_1$. 
			Then let $I'$ be the subsequence of $I$ consisting of the indices $j$ with $\bar P_{j} \neq \emptyset$.
			If $I=I'$, then $\contractt{H}$ has a $\bar u$-$\bar v$ path consisting of $s$ vertices. 
			Otherwise, since the internal vertices of $\tau_1$ has degree at least two in $\contract{H}$, 
			the vertices of $\contractt{H}$ corresponding to $\bar P(i)$ and $\bar P(j)$ are connected by an edge
			for every consecutive indices $i$ and $j$ of $I'$. Therefore, there exists a $\bar u$-$\bar v$ path in $\contractt{H}$
			consisting of $t$ vertices.
			
			\item \textbf{Case 2. There is a $u$-$v$ path $\tau'$ in $\contract{G}$ with no internal vertices in $U$.} 
			Recall that $V(\contract{G})=V(\contract{H})$. 
			If $\tau'$ entirely exists in $\contract{H}$, there is an edge $(\bar u,\bar v)$ in $\contractt{H}$, and thus we are done. 
			If it is not the case, let $(u',v')$  be an edge in $\tau'$ which does not exist in $\contract{H}$. 
			By Lemma~\ref{spanner}, there is a $u'$-$v'$ path in $\contract{H}$ consisting of $t$ vertices. 
			Since the internal vertices of $\tau'$ have degree exactly two, the $u'$-$v'$ path contains both $u$ and $v$.
			Therefore, the subgraph of the $u'$-$v'$ path lying between $u$ and $v$ is a $u$-$v$ path in $\contract{H}$  consisting of at most $t$ vertices. 
			Then similarly to Case~1, we can show that there is a $\bar u$-$\bar v$ path in $\contractt{H}$ consisting of $t$ vertices. 
		\end{itemize}
		
		Therefore, for an edge $(\bar u,\bar v)$ in $\contractt{G}$, there is a $\bar u$-$\bar v$ path in $\contractt{H}$ consisting of at most $t$ vertices. This completes the proof.
	\end{proof}
	
	\begin{lemma} \label{H2ptoG2p}
		The weight of each connected component of $\contractt{G}-S'$ is at most $2/3$ of the weight of $\contractt{G}$.
	\end{lemma}
	\begin{proof}
		%
		We show that
		each connected component of $\contractt{G}-S'$ is a subset of a connected component of $\contractt{H}-S$.
		This implies that 		each connected component of $\contractt{G}-S'$ consists of at most $2|V(\contractt{G})|/3$ vertices
		 since $S$ is a balanced separator of $\contract{G}$, and $V(\contractt{G})=V(\contractt{H})$.

		To show this, consider an edge $(a, b)$ of $\contractt{G}-S'$. We claim that 
		$a$ and $b$ are contained in the same connected component of $\contractt{H}-S$. 
		By Lemma~\ref{spanner2}, 
		there is an $a$-$b$ path in $\contractt{H}$ consisting of $t$ vertices. 
		If this path contains a vertex of $S$, then $a$ and $b$ are contained in $S'$ by construction, which makes a contradiction. 
		Therefore, an $a$-$b$ path in $\contractt{H}$ does not intersects $S$, which means that $a$ and $b$ are contained in 
		the same connected component of $\contractt{H}-S$.
		
		Therefore, each connected component of $\contractt{G}-S'$ is a subset of the connected component of $\contractt{H}-S$, and the lemma holds. 
	\end{proof}
	
	\subparagraph{Tree decomposition of $\contractt{G}$.}
	To construct a tree decomposition of $\contractt{G}$, we observe that any subgraph of $\contractt{G}$ also has a balanced separator of small weight. 
	Let $A$ be a connected component of $\contractt{G}-S'$.
	Then consider the subgraph $\contractt{H}[A]$ of $\contractt{H}$ induced by the vertices in $A$.
	Since $\contractt{H}[A]$ is planar, it has a balanced separator with small $h$-weight, and thus $A$ has a balanced separator with small $g$-weight.
	Therefore, we can recursively compute a balanced separator for every connected component of $\contractt{G}-S'$. 

	We compute a tree decomposition of each connected component of $\contractt{G}-S'$ 
	recursively, and then we connect those trees by one additional empty bag, and add the separator to all bags of the resulting tree.
	Since the weight of the connected components decreases geometrically, the maximum weight of the bags is $O(\sqrt{|U|})$.
	Therefore, we can compute a tree decomposition $(T,\beta)$ of $\contractt{G}$  of weighted width $O(\sqrt{|U|})$.
	
	\subparagraph{Tree decomposition of $\contract{G}$.}
	The remaining step is computing a tree decomposition of $\contract{G}$ from $(T,\beta)$.
	Then we start with the pair $(T, \beta)$, and then add several nodes (and bags) to $T$ and add several vertices to the bags of $\beta$
	as follows. 
	
	Let $Q$ be the set of vertices $u$ of $\contract{G}$ with $P(u)\cap U\neq \emptyset$, and let 
	$Q^\mathsf{c}$ be the set of vertices of $\contract{G}$ not in $Q$. 
	For any vertex $v$ of $Q$, there is a vertex $\bar v$ in $\contractt{G}$ with $P(v)\cap U=\bar P( \bar v)$,
	and thus $\bar v$ is contained bags of $(T,\beta)$. We replace $\bar v$ with $v$ in the bags of $(T, \beta)$ containing $v$. 

	On the other hand, for a vertex $v$ of $Q^\mathsf{c}$, $P(v) \cap U =\emptyset$, and thus no vertex of $\contractt{G}$ corresponds to $v$.
	Thus we are required to insert $v\in Q$ into $T'$. 
	Recall that every vertex in $Q$ has degree one or two in $G$, and thus it also has degree one or two in $\contract{G}$. Also, $\contract{G}$ is connected.\footnote{We assume that the input intersection graph $G$ is connected, and thus $\contract{G}$ is connected by construction.}
	Therefore, $\contract{G}[Q]$ consists of a number of simple paths and isolated vertices, all of which are connected to at most two vertices of $Q$. Thus, there are the following four cases.

	\begin{itemize}
		\item \textbf{Case 1.} An isolated vertex $v$ is connected to only one vertex $u\in Q$.
		For a bag $\beta(x)$ of $(T,\beta)$ containing $u$, we add a leaf with bag $\{u,v\}$ as a child of $x$. 
		
		\item \textbf{Case 2.} An isolated vertex $v$ is connected to exactly two vertices $u, w\in Q$.
		By construction, there is an edge $(u,w)$ in $\contractt{G}$. 
		Then, there is a bag $\beta(x)$ of $T$ containing both $\bar u$ and $\bar w$.  
		We add  a leaf with bag $\{u,v,w\}$ as a child of $x$. 
		
		\item \textbf{Case 3.} A simple $v$-$v'$ path $(v_1,...,v_k)$ is connected to exactly one vertex $u\in Q$.
		Note that either $v$ or $v'$ is connected to $u$, but not both. 
		Without loss of generality, we assume that there is an edge $(u,v)$ in $\contract{G}$.
		For a bag $\beta(x)$ of $T$ containing $u$, we add  a node $x_1$ with bag $\{u,v\}$ as a child of $x$. 
		For $1< t < k$, 
		we add a node $x_t$ with bag $\{v_t, v_{t+1}\}$ as a child of $x_{t-1}$.
		
		\item \textbf{Case 4.} A simple $v$-$v'$ path $(v_1,...,v_k)$ is connected to two vertices $u, w\in Q$.
		In this case, $v$ is connected to one of $u$ and $w$, say $u$, and $v'$ is connected to the other vertex, say $w$. 
		Since $\contractt{G}$ has an edge $(\bar u, \bar w)$, there is a bag $\beta(x)$ of $T$ that contains both $u$ and $w$, where
		$\bar u$ and $\bar v$ are the vertices of $\contractt{G}$ corresponding to $u$ and $v$, respectively.
		We add a node $x_1$ with bag $(u, w, v, v')$ as a child of $x$.
		For $1< t < k/2$, 
		we add a node $x_t$ with bag $\{v_t, v_{t+1}, v_{k-t}, v_{k-t+1}\}$ as a child of $x_{t-1}$.
	\end{itemize}
	
	In this way, $(T,\beta)$ is a tree decomposition of $\contract{G}$.
	Now we analyze the weight of a bag of $(T,\beta)$. 
	Since we replace $\bar v$ with $v$ for a vertex $v$ of $Q$ in the bags of $(T,\beta)$,
	the weight of a bag $B$ of $(T,\beta)$ consisting of vertices of $Q$ increases. 
	Let $v$ be a vertex of $Q$, and let $\bar v_i$ be the vertex of $\contractt{G}$ corresponding to $v$. 
	Since $P(v)$ is composed of at most $O(1)$ cliques, it contains at most $O(1)$ vertices of $V\setminus U$. 
	Note that $B$ contains $O(\sqrt{|U|})$ vertices  because the weighted width of the tree decomposition of $\contractt{G}$ is $O(\sqrt{|U|})$. 
	Therefore, the weight of $B$ increases by $O(\sqrt{|U|})$ after replacing vertices of $\contractt{G}$ with vertices of $\contract{G}$. 
	Thus, the weight of a bag consisting of vertices of $Q$ is  $O(\sqrt{|U|})$. 
	
	Now we show that the weight of a bag containing a vertex of $Q^{\mathsf c}$ is  $O(\sqrt{|U|})$. 
	Every vertex $v\in Q^{\mathsf{c}}$ is a contraction of simple paths or simple cycles of length at most $O(1)$ in $G$. 
	In particular, its weight in $\contract G$ is $O(1)$. 
	Since a bag containing a vertex of $Q^{\mathsf{c}}$ consists of at most two vertices of $Q^{\mathsf{c}}$ and at most two vertices
	of $Q$, the weight of the bag is $O(\sqrt{|U|})$.

	\subsection{Computing a Tree Decomposition of Small Weighted Width}\label{sec:construction}
We can compute a tree decomposition of weighted treewidth $O(\sqrt{|U|})$ of $\contract{G}$ using the 
algorithm proposed by de Berg et al.~\cite{de2020framework},
which computes a tree decomposition of weighted width $O(w)$ of $\contract{G}$ assuming that the weighted treewidth of $\contract{G}$ is $w$. 
This algorithm runs in $2^{O(w)}(n+m)$ time. 
This algorithm works on the graph itself even if  the geometric representation is unknown. 
In this section, we briefly describe how the algorithm proposed by de Berg et al.~\cite{de2020framework} for computing such a tree decomposition
works.

%

A $\emph{blowup}$ of a vertex $v$ by an integer $t$ results in a graph where we replace a single vertex $v$ into a clique of size $t$, in which 
we connect every vertex to the neighbor vertices of $v$. We denote the set of vertices in the clique by $\mathcal{B}(v)$.
For the graph $\contract{G}$, we construct an unweighted graph $\contractb{G}$ by blowing up each vertex $v$ by an integer $\lceil g(v)\rceil$ where $g(v)=\log{|P(v)|}+1$.
We compute a tree decomposition of $\contractb{G}$, denoted by $(T_{\mathcal B}, \sigma_{\mathcal B})$.
Then we compute a tree decomposition of $\contract{G}$ by using the same tree layout of $T_{\mathcal B}$.
In particular, we add a vertex $v\in \contract{G}$ to a bag if and only if the bag in $T_{\mathcal B}$ contains all 
vertices of $\mathcal{B}(v)$.
The Lemma~2.9 and Lemma~2.10 of~\cite{de2020framework} prove the correctness of this algorithm.

The remaining step is computing a tree decomposition of $\contractb{G}$.
We apply a constant-factor approximation algorithm for computing an optimal tree decomposition~\cite{PBook}. 
By Theorem~\ref{pc_wtw} and Lemma 2.9 of~\cite{de2020framework}, the treewidth of $\contractb{G}$ is $O(\sqrt{|U|})$.
Therefore, the algorithm returns a tree decomposition of $\contractb{G}$ of width  $O(\sqrt{|U|})$. 
This algorithm takes a $2^{O(tw(\contractb{G}))}n=2^{O(\sqrt{|U|})}(n+m)$ time, and therefore, 
we can compute a tree decomposition of $\contract{G}$ of weighted width $O(\sqrt{|U|})$ without using the geometric representation of the graph.
This completes the proof of Theorem~\ref{pc_wtw}.

\section{Computing a Feedback Vertex Set}
In this section, we present a $2^{O(\sqrt{k})}(n+m)$-time algorithm for finding a feedback vertex set of size $k$ for the
intersection graph $G[F]=(V,E)$ of similarly sized fat objects in the plane. 

To do this, we first show that $G$ has $O(k)$ vertices of degree at least three if it has a feedback vertex set of size at most $k$ in
Section~\ref{sec:tw-yes}. 
Therefore, by Theorem~\ref{pc_wtw}, 
for a greedy partition $\mathcal{P}$ of $G$, 
the weighted treewidth of $\contract{G}$ is $O(\sqrt{k})$ if $(G,k)$ is a \textbf{yes}-instance.

To use this fact, we first compute such a 
greedy partition $\mathcal{P}$ in $O(n+m)$ time using the algorithm in~\cite{de2020framework}.
Then we check if the weighted treewidth of $\contract{G}$ is $O(\sqrt{k})$ in $2^{O(\sqrt{k})}(n+m)$ time.
If it is $\omega(\sqrt{k})$, we conclude that $(G,k)$ is a \textbf{no}-instance immediately. 
Otherwise, we compute a feedback vertex set of size $k$, if it exists, using dynamic programming on a
tree decomposition of $\contract{G}$ of weighted width $O(\sqrt{k})$. 
De Berg et al.~\cite{de2020framework} presented a dynamic programming algorithm
for \fvs which runs in $2^{O(\sqrt{w})}(n+m)$ time, where $w$ is the weighted treewidth of  $\contract{G}$.
We can use this algorithm, but for completeness, we briefly describe this algorithm in Section~\ref{sec:dp}.

\subsection{Weighted Treewidth for a \textbf{Yes}-Instance}\label{sec:tw-yes}

We first show that $G$ has $O(k)$ vertices of degree at least three if it has a feedback vertex set of size $k$. 
We consider a partition of the plane into axis-parallel squares (called \emph{cells}) with diameter $1/\sqrt{2}$,
which we call the \emph{grid}. Recall that we let the smallest diameter of the objects defining $G$ be one.
For a cell $\sigma$, let $P(\sigma)$ be the set of the centers of the objects defining $G$. 
By definition, $P(\sigma)$ forms a clique in $G$. 
We say a grid cell is a \emph{neighbor} of another grid cell if the smallest distance between two points from the two cells is at most $2\alpha\gamma$,
where $\alpha$ is the measure for the fatness of the objects of $F$, and $\gamma$ is the ratio of the largest and the smallest diameter of the objects. 
We say a grid cell $\sigma$ is \emph{heavy} if $P(\sigma)$ consists of at least three points, and \emph{light}, otherwise. 
Let $P_h$ denote the set of vertices contained in the
heavy cells, and let $P_\ell$ denote the set of vertices contained in the light cells. 

The following lemma implies that $\contract{G}$ has a weighted treewidth of $O(\sqrt{k})$ 
if it has a feedback vertex set of size $k$. 
The proof of the following lemma is inspired by the proof of~{[Lemma 9.1, \cite{PBook}]}.
Let $G'$ be the graph obtained from $G$ by removing all degree-1 vertices. 
The weighted treewidth of $G'$ is at most the weighted treewidth of $G$. 
Also, $(G,k)$ is a \textbf{yes}-instance of \textsc{feedback vertex set} if and only if
$(G',k)$ is a \textbf{yes}-instance of \textsc{feedback vertex set}.

\begin{lemma}\label{lem:fvs-num}
	$G'$ has $O(k)$ vertices of degree at least three if $(G',k)$ is a \textbf{yes}-instance of \textsc{feedback vertex set},
	where $G'$ is the graph obtained from $G$ by removing all degree-1 vertices. 
\end{lemma}
\begin{proof}
	If $G'$ has a feedback vertex set of size $k$, the number of heavy cells is at most $k$, and the size of $P_h$ is at most $3k$.
	This is because all but two of the vertices contained in the same cell must be contained in a feedback vertex set of $G$.  
	In the following, we show that the number of vertices of $P_\ell$ of degree at least three is at most $O(k)$. 
	
	To do this, consider the subgraph $G_\ell$ of $G'$ induced by $P_\ell$. We first observe that every vertex of $G_\ell$ has degree at most $4(\alpha\gamma)^2$ in $G_\ell$. (But note that its degree in $G'$ might be larger than $4(\alpha\gamma)^2$.) Moreover, $G_\ell$ also has a feedback vertex set of size $k$. 
	Let $X$ be a feedback vertex set of $G_\ell$ of size $k$, and let $Y$ be the set of vertices of $G_\ell$ not contained in $X$. 
	By definition, the subgraph $G[Y]$ of $G_\ell$ induced by $Y$ is a forest. 
	The number of vertices of degree at least three is linear in the number of leaf nodes of the trees.  
	Since $G'$ has no degree-1 vertex, 
	a leaf node of $G[Y]$ is incident to a vertex of $X\cup P_h$ in $G'$. 
	Each vertex in $X\cup P_h$ is incident to at most $4(\alpha\beta)^2$ vertices of $Y$. 
	Since $|X\cup P_h|$ is $O(k)$, the total number of vertices of leaf nodes is $O(k)$,
	and thus, the total number of vertices of degree at least three of $G_\ell$ (and $G'$) is $O(k)$. 
\end{proof}

\begin{lemma}
	For a greedy partition $\mathcal P$, the weighted treewidth of $\contract{G}$ is $O(\sqrt{k})$  if $(G,k)$ is a \textbf{yes}-instance of \textsc{feedback vertex set}.
\end{lemma}

\subsection{Dynamic Programming}\label{sec:dp}
By combining Lemma~\ref{lem:fvs-num} and Theorem~\ref{pc_wtw}, we can obtain a tree decomposition of weighted width $O(\sqrt{k})$ if a geometric intersection graph $G$
has a feedback vertex set of size $k$. De Berg et al.~\cite{de2020framework} showed how to compute a feedback vertex set of size $k$
in $2^{O(\sqrt{w})} \poly n$ time, where $w$ denotes the weighted treewidth of $\contract{G}$. 
In this section, we briefly describe how this algorithm works. 

Given a tree decomposition $(T,\beta)$ of a graph $G'$, algorithm for solving \textsc{Feedback Vertex Set} 
computes for each bag $\beta(t)$ and each subset $S\subseteq \beta (t)$ together with the connectivity information $\eta$ of the vertices of $S$, the 
maximum size $c[t, S, \eta]$ of a feedback vertex set $\hat{S}$ of $G[t]$ with $\hat{S} \cap \beta(t) = S$ satisfying the connectivity information $\eta$, 
where $G[t]$ denotes the subgraph of $G$ 
induced by the vertices appearing in bags in the subtree of $T$ rooted at $t$. 

Let $t$ be a node of a tree decomposition $(T,\beta)$ of $\contract{G}$. Recall that 
each vertex $v$ of $\beta(t)$ corresponds to the partition class $P(v)$. Let $X_t = \bigcup_{v\in \beta(t)} P(v)$.
A feedback vertex set contains at least $|C|-2$ vertices of a clique $C$. Therefore, since from each partition class 
we can exclude $O(1)$ vertices (at most two vertices from each clique), the number of subsets
$\hat{S}$ that need to be considered is at most 
\[\prod_{v\in \beta (t)} |P(v)|^2 = \exp\biggl( \sum_{C\in \beta (t)} 2\log |P(v)|\biggr) = 2^{O(\sqrt{k})}.\] 
Therefore, we can track connectivity of these subsets by applying the rank-based approach of~\cite{BODLAENDER201586}, which allows us to keep the number of equivalence relations
considered single-exponential in $O(\sqrt{k})$.  
Assuming that we have $c[t', \cdot,\cdot]$ for all descendants $t'$ of $t$ in $T$ and their connectivity information, we can compute $c[t, S, \eta]$ in $2^{O(\sqrt{k})}$ time.
Since the number of nodes in the tree decomposition is $O(kn)$, the total running time of the dynamic programming algorithm is 
$2^{O(\sqrt{k})}(n+m)$. 
Therefore, we have the following theorem. 
\begin{theorem}
	Given a intersection graph $G$ of similarly sized fat objects in the plane (without its geometric representation), we can compute a feedback vertex set of size $k$ in $2^{O(\sqrt {k})}(n+m)$ time, if it exists.
\end{theorem}

	\newpage
	\bibliography{paper}{}
	
\end{document}